\documentclass[journal,draftcls,onecolumn,12pt]{IEEEtran}

\usepackage{amsmath,graphicx,amsthm}
\usepackage{comment,amsmath,amssymb,epsfig,amsfonts,verbatim, subfigure}
\usepackage{psfrag}
\usepackage{graphicx,times}
\usepackage{algorithmic}
\usepackage[boxed]{algorithm}
\usepackage{cite}
\usepackage{units}
\usepackage{color}
\usepackage{mathrsfs}
\usepackage{stfloats}
\usepackage{subfigure}
\usepackage{fixltx2e}

\graphicspath{{./Figures/}}

\newtheorem{theorem}{Theorem}

\newtheorem{lemma}{Lemma}

\theoremstyle{definition}

\def\gg{\left|g\right|^2}

\def\fg{\left|f\right|^2}
\def\dg{\left|d\right|^2}

\def\C{{\mathcal C}}
\def\N{{\mathcal N}}
\def\E{{\mathcal E}}

\title{Benefits of improper signaling for underlay cognitive radio}

\author{\IEEEauthorblockN{Christian Lameiro,~\IEEEmembership{Student Member,~IEEE,} Ignacio Santamar\'ia,~\IEEEmembership{Senior~Member,~IEEE,} and Peter J. Schreier,~\IEEEmembership{Senior~Member,~IEEE}\\}
\thanks{C. Lameiro and I. Santamar\'ia are with the Department of Communications Engineering,
University of Cantabria, Spain (e-mail: \{lameiro, nacho\}@gtas.dicom.unican.es). P. J. Schreier is with the Signal \& System Theory Group, Universit\"{a}t Paderborn, Germany (email: peter.schreier@sst.upb.de).}
\thanks{C. Lameiro and I. Santamar\'ia have received funding from the Spanish Government (MICINN) under projects CONSOLIDER-INGENIO 2010 CSD2008-00010 (COMONSENS), TEC2013-47141-C4-3-R (RACHEL) and FPU Grant AP2010-2189; and also from the Deutscher Akademischer Austauschdienst (DAAD) under its programm "Research grants for doctoral candidates and young academics and scientists". P. Schreier receives financial support from the Alfried Krupp von Bohlen und Halbach foundation, under its program "Return of German scientists from abroad".}}

\begin{document}
\maketitle
\begin{abstract}
In this letter we study the potential benefits of improper signaling for a secondary user (SU) in underlay cognitive radio networks. We consider a basic yet illustrative scenario in which the primary user (PU) always transmit proper Gaussian signals and has a minimum rate constraint. After parameterizing the SU transmit signal in terms of its power and circularity coefficient (which measures the degree of impropriety), we prove that the SU improves its rate by transmitting improper signals only when the ratio of the squared modulus between the SU-PU interference link and the SU direct link exceeds a given threshold. As a by-product of this analysis, we obtain the optimal circularity coefficient that must be used by the SU depending on its power budget. Some simulation results show that the SU benefits from the transmission of improper signals especially when the PU is not highly loaded.
\end{abstract}
\section{Introduction}
\label{sec:intro}
Interference management techniques, the main limiting factor in current wireless networks, have received a lot of attention in recent years. One of the most interesting results is regarding the statistical nature of the transmitted signals in an interference-limited network. In such scenarios, the use of proper complex Gaussian signals has typically been assumed due to the fact that these are capacity achieving in the point-to-point, broadcast and multiple access channels \cite{Cover}. However, some recent results have proven that improper complex Gaussian signals increase the achievable rates in various interference-limited networks \cite{Ho2012,Zeng2013,Cadambe2010,ICC2013}, hence calling into question the widely-used assumption of proper Gaussian signals when interference presents a limiting factor. Improper signals have real and imaginary parts that have unequal power and/or are correlated. They have been successfully applied to enlarge the achievable rate region of the K-user multiple-input multiple-output interference channel (MIMO-IC) \cite{Zeng2013}, as well as to provide additional degrees of freedom in some specific scenarios \cite{Cadambe2010,ICC2013}.

At the same time, cognitive radio (CR) has emerged as a paradigm for efficient utilization of radio resources \cite{Mitola1999}. The fundamental idea of CR is to let the so-called cognitive users intelligently use side information for spectrum sharing. Typically, these users are called unlicensed or secondary users (SUs), which coexist with licensed or primary users (PUs) without disrupting the PUs' communication. Three main CR schemes have been proposed, which require different cognition levels: interweave, overlay, and underlay \cite{Goldsmith2009}. In this letter, we consider underlay CR, where the SUs are allowed to coexist with the PUs as long as they guarantee the PUs a given performance metric. Underlay CR has been studied for a wide range of scenarios \cite{Yang2013,Cumanan2012,LameiroWSA2014,Jorswieck2011}. Typically, optimization of the CR network (power control, beamforming, etc.) is carried out under interference temperature constraints to protect the PUs. On the other hand, other works have considered a rate constraint at the PU \cite{Cumanan2012}, or use spatial shaping constraints when the SU is equipped with multiple antennas \cite{LameiroWSA2014}.

In this letter we study whether improper signaling can be beneficial in underlay CR. More specifically, we consider a basic yet illustrative CR scenario, in which one single-antenna SU wishes to access the channel in the presence of a single-antenna primary point-to-point link that has a rate constraint. Since the PU is typically unaware of the SU and thus no cooperation exists between them, we assume that the PU always transmits proper Gaussian signals (i.e., its transmission strategy is independent of that of the SU), whereas the SU may transmit either proper or improper Gaussian signals depending on which signaling scheme performs better. We analyze this setting and prove that the SU improves its rate by transmitting improper signals only when the ratio of the squared modulus between the SU-PU interference link and the SU direct link exceeds a given threshold. To the best of our knowledge, this is the first work that tackles the problem of improper signaling in CR scenarios.


The remainder of this letter is organized as follows. Section \ref{sec:sec1} reviews some properties of improper random variables and introduces the system model. Section \ref{sec:rates} derives the achievable rate of the SU for the proper and improper cases. In Section \ref{sec:results} we provide some numerical examples to illustrate our findings. Section \ref{sec:conclusions} presents the concluding remarks.

\section{Improper signaling in underlay cognitive radio}\label{sec:sec1}
\subsection{Preliminaries of improper random variables}
Here we provide the main definitions and results on improper random variables that will be used throughout the paper. For a comprehensive analysis of improper signals, we refer the reader to \cite{Schreier2010}.

The \emph{complementary-variance} of a zero-mean complex random variable $x$ is defined as $\tau_x\doteq\E\left\{x^2\right\}$, where $\E\{\cdot\}$ is the expectation operator. If $\tau_x=0$, then $x$ is called \emph{proper}, otherwise \emph{improper}.

The circularity coefficient of a complex random variable $x$ is defined as the absolute value of the quotient of its complementary-variance and variance, i.e.,
\begin{equation}
    \kappa_x\doteq\frac{\left|\tau_x\right|}{\sigma_x^2} \; .
\end{equation}
The circularity coefficient satisfies $0\leq\kappa_x\leq1$ and measures the degree of impropriety of $x$. If $\kappa_x=1$ we call $x$ \emph{maximally improper}.

\subsection{System description}\label{sec:model}
We consider a CR scenario, modeled as a two-user single-input single-output interference channel (SISO-IC), where a SU wishes to access the channel in presence of a PU. As depicted in Fig. \ref{fig:IC}, the PU transmits with fixed power, $p\leq P$, and has a minimum rate constraint, $\bar{R}$; whereas the SU must adjust its transmit power, $q\leq Q$, in order to control the interference level at the primary receiver such that its rate requirement is guaranteed. In this setting, the signal received by the primary and secondary receivers can be expressed, respectively, as
\begin{align}
    y_p&=h\sqrt{p}s_p+g\sqrt{q}s_s+n_p \; , \\
    y_s&=f\sqrt{q}s_s+d\sqrt{p}s_p+n_s \; ,
\end{align}
where $h$, $d$, $g$ and $f$ are the PU-PU, PU-SU, SU-PU and SU-SU channels (see Fig. \ref{fig:IC}); $n_p\sim\mathcal{CN}(0,\sigma_p^2)$ and $n_s\sim\mathcal{CN}(0,\sigma_s^2)$ are the noise at the primary and secondary receivers, respectively, where $\mathcal{CN}(0,\sigma^2)$ denotes a proper complex Gaussian distribution with zero mean and variance $\sigma^2$; and $s_p$ and $s_s$ are the Gaussian transmitted symbols, with $\mathcal{E}\{\left|s_p\right|^2\}=\mathcal{E}\{\left|s_s\right|^2\}=1$. For the sake of simplicity, we assume henceforth $\sigma_p=\sigma_s=\sigma$. Our results can be easily extended to different noise variances. In general terms, the signal transmitted by the SU can be parameterized in terms of its power, $q$, and the circularity coefficient, $\kappa$, which measures the degree of impropriety. The rate achieved by the PU can be written in terms of the circularity coefficient as \cite[eq. (30)]{Zeng2013}
\begin{equation}\label{eq:RpuImprop}
    R_{PU}\left(q,\kappa\right)=\log_2\left(1+\frac{p\left|h\right|^2}{\sigma^2+q\left|g\right|^2}\right)+\frac{1}{2}\log_2\frac{1-\kappa_{y_p}^2}{1-\kappa_{in_p}^2} \; ,
\end{equation}
where $\kappa_{y_p}$ and $\kappa_{in_p}$ are the circularity coefficients of the received and interference-plus-noise signals at the PU, respectively, which are given by
\begin{align}
    \kappa_{y_p}&=\frac{\kappa}{1+\frac{p\left|h\right|^2+\sigma^2}{q\left|g\right|^2}} \; , \\
    \kappa_{in_p}&=\frac{\kappa}{1+\frac{\sigma^2}{q\left|g\right|^2}} \; .
\end{align}
Taking $\kappa=0$ in \eqref{eq:RpuImprop} yields the well-known expression for the proper signaling case. Furthermore, using \cite[eq. (30)]{Zeng2013} the achievable rate of the SU as a function of $q$ and $\kappa$ can be expressed as
\begin{equation}\label{eq:Rsu}
    R_{SU}\left(q,\kappa\right)=\frac{1}{2}\log_2\left\{\frac{q\left|f\right|^2}{\tilde{\sigma}^2}\left[\left(1-\kappa^2\right)\frac{q\left|f\right|^2}{\tilde{\sigma}^2}+2\right]+1\right\} \; ,
\end{equation}
where $\tilde{\sigma}^2=\sigma^2+p\dg$ is the interference-plus-noise power at the secondary receiver.
\begin{figure}[t]
\centering
\includegraphics[width=0.85\columnwidth]{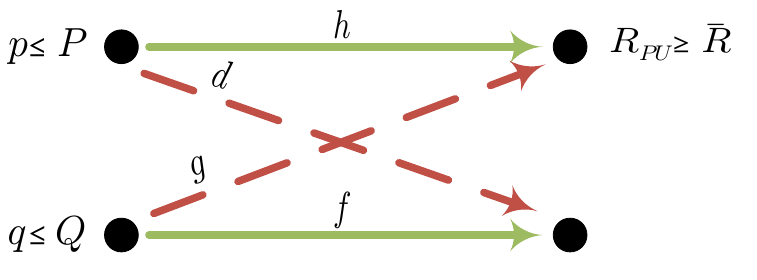}
\caption{A simple underlay CR scenario modeled as a two-user SISO-IC. The SU (bottom link) may transmit improper signals, but must guarantee the rate constraint of the PU (top link).}
\label{fig:IC}
\end{figure}

From these expressions it becomes evident that, if $q$ is kept fixed, increasing the degree of impropriety of the transmitted signal will clearly decrease the achievable rate of the SU but, at the same time, it will also increase the rate of the PU, thus allowing the SU to increase its transmit power while keeping the rate of the PU above its requirement. It is therefore clear that there exists a tradeoff between the additional transmit power that must be used for the SU to maintain its proper signaling rate, and the extra power that it is allowed to transmit. Improper signaling will be beneficial only when the latter is greater than the former. In the next section we derive the achievable rate of the SU as a function of $\kappa$, which will shed light onto this tradeoff.

\section{Achievable rates}\label{sec:rates}
In this section we derive the achievable rate of the SU for both cases, i.e., proper and improper transmissions, when the PU data rate is constrained as $R_{PU}(q,\kappa)\geq\bar{R}$, by expressing the allowable transmit power, $q$, as a function of $\kappa$. We express the rate constraint as a fraction of the maximum achievable rate (in the absence of interference), i.e., $\bar{R}=\alpha R_{PU}(0,0)=\alpha\log_2\left(1+\frac{p\left|h\right|^2}{\sigma^2}\right)$, where $\alpha\in\left[0,1\right]$ is the loading factor. Notice that, since setting $\kappa=0$ yields a proper Gaussian signal, the results obtained for the improper case can be specialized to the proper one. However, for the sake of exposition, we provide separate expressions for both cases.
\subsection{Proper signaling case}
When $\kappa=0$, the achievable rate of the SU can be expressed as
\begin{equation}\label{eq:RsuProp}
    R_{SU}\left(\kappa=0\right)=\log_2\left(1+\frac{q\left(\kappa=0\right)\left|f\right|^2}{\tilde{\sigma}^2}\right) \; ,
\end{equation}
$q(\kappa=0)$ is the allowable power for the proper case, which is obtained by equating \eqref{eq:RpuImprop} to $\alpha R_{PU}(0,0)$ with $\kappa=0$, and is given by
\begin{equation}
    q(\kappa=0)=\frac{\sigma^2}{\left|g\right|^2}\left(\frac{\gamma\left(1\right)}{\gamma\left(\alpha\right)}-1\right) \; ,
\end{equation}
where $\gamma(a)=2^{aR_{PU}(0,0)}-1$, i.e., the required signal-to-noise ratio (SNR) needed to achieve a rate $aR_{PU}(0,0)$ in the absence of interference. Note that we have dropped the dependence with $q$ in \eqref{eq:RsuProp}, since the rate of the SU is now a function of $\kappa$ only.

\subsection{Improper signaling case}
Analyzing the achievable rate of the SU as a function of $\kappa$ will provide us with insights on the properties of improper signaling for this scenario. To this end, we provide the following lemma.
\begin{lemma}\label{prop:Rsu}
    When the rate of the PU is constrained as $R_{PU}\geq\alpha R_{PU}(0,0)$, the achievable rate of the SU can be expressed in terms of its circularity coefficient as
    \begin{align}\label{eq:RsuProp1}
        R_{SU}\left(\kappa\right)=\frac{1}{2}\log_2\left[\frac{2q\left(\kappa\right)\left|f\right|^2}{\tilde{\sigma}^2}\left(1-\beta\frac{\left|f\right|^2\sigma^2}{\left|g\right|^2\tilde{\sigma}^2}\right)+\frac{\left|f\right|^4\sigma^4}{\left|g\right|^4\tilde{\sigma}^4}\left(\frac{\gamma(2)}{\gamma(2\alpha)}-1\right)+1\right] \; ,
    \end{align}
    where $q\left(\kappa\right)$ is the allowed transmit power, which is given by
    \begin{align}\label{eq:qsuProp1}
        q\left(\kappa\right)=\left[\sqrt{\beta^2+\left(1-\kappa^2\right)\left(\frac{\gamma(2)}{\gamma(2\alpha)}-1\right)}-\beta\right]\frac{\sigma^2}{\left|g\right|^2\left(1-\kappa^2\right)} \; ,
    \end{align}
    and $\beta=1-\frac{\gamma\left(1\right)}{\gamma\left(2\alpha\right)}$ is a parameter that satisfies $\beta \leq 1$.
\end{lemma}
\begin{proof}
    Please refer to Appendix \ref{app:A}.
\end{proof}
Lemma \ref{prop:Rsu} provides an alternative expression for the SU rate that depends on $\kappa$ only through its impact on the allowed transmit power, $q\left(\kappa\right)$. With this observation, we obtain the following result.
\begin{theorem}\label{le:ImpropCond}
    When the rate of the PU is constrained as $R_{PU}\geq\alpha R_{PU}(0,0)$ and $q(0)<Q$, the achievable rate of the SU is improved by transmitting improper signals if and only if
    \begin{equation}\label{eq:ImpropCond}
        \frac{\left|g\right|^2\tilde{\sigma}^2}{\left|f\right|^2\sigma^2}>\beta \; ,
    \end{equation}
    with $\beta=1-\frac{\gamma\left(1\right)}{\gamma\left(2\alpha\right)}$. Furthermore, when \eqref{eq:ImpropCond} holds, the optimal value of $\kappa$ is given by
    \begin{equation}\label{eq:kappaOpt}
        \kappa^\star=\left\{\begin{array}{cc}
        1 & \text{if} \; q\left(1\right)\leq Q \\
        \sqrt{1-\frac{\sigma^2}{Q\gg}\left[\left(\frac{\gamma(2)}{\gamma(2\alpha)}-1\right)\frac{\sigma^2}{Q\gg}-2\beta\right]} & \text{otherwise}\end{array}\right. \; .
    \end{equation}
\end{theorem}
\begin{proof}
    Please refer to Appendix \ref{app:B}.
\end{proof}
This theorem provides a necessary and sufficient condition for improper signaling to be beneficial in the considered scenario, which, thanks to its simplicity, provides interesting insights. We observe that, since $\beta \leq 1$ and $\tilde{\sigma}^2=\sigma^2+p\dg \geq \sigma^2$, if the gain of the interfering channel, $g$, is greater than that of the SU direct channel, $f$, then the use of improper signaling will always enhance the SU data rate independently of the rate constraint and SNR of the PU. Alternatively, when the SNR of the PU, $\gamma(1)$, is equal to or greater than $\gamma(2\alpha)$ (i.e., $\alpha\leq0.5$ and hence the PU can achieve its rate constraint by using only the real or imaginary part of the transmitted symbol), improper signaling also improves the SU rate independently of the interfering and SU direct channels. Moreover, it is shown that, if improper signaling is beneficial, then maximally improper signals are optimal. Note, however, that the optimal transmitted signals may not always be chosen as maximally improper due to the limited power budget at the SU. In those cases, the circularity coefficient must be the minimum value that allows the SU to transmit with its maximum power. Condition \eqref{eq:ImpropCond} shows that improper signaling is beneficial if the increase in allowable power due to improper transmissions is greater than the additional power needed by the SU to maintain the rate achieved by proper signaling.

\section{Numerical examples}\label{sec:results}
\begin{figure}[t]
\centering
\includegraphics[width=0.85\columnwidth]{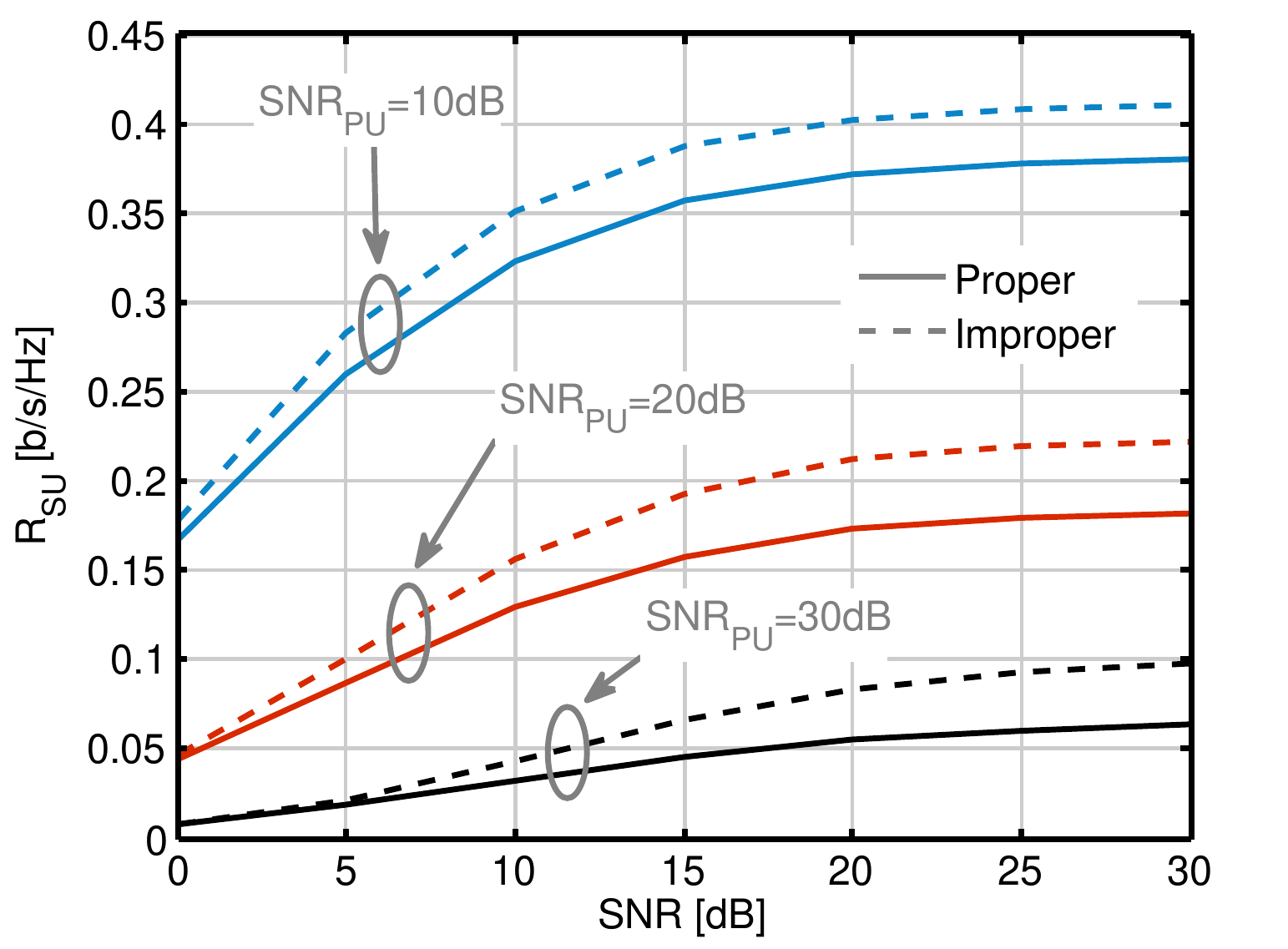}
\caption{Achievable rate of the SU as a function of its SNR, for different SNRs of the PU and $\alpha=0.8$.}
\label{fig:ICalpha08}
\end{figure}
In this section we provide some numerical examples that illustrate our findings. Assuming Rayleigh fading, all channel coefficients are distributed as $\C\N\left(0,1\right)$. The SNR of the PU and SU is respectively defined as $\text{SNR}_{PU}=\frac{P}{\sigma^2}$ and $\text{SNR}_{SU}=\frac{Q}{\sigma^2}$. Since the PU has no additional constraints, we set $p=P$, i.e., it transmits with its maximum available power. Without loss of generality, we consider $\sigma^2=1$. All results are averaged over $10^5$ independent channel realizations.

Fig. \ref{fig:ICalpha08} shows the achievable rate of the SU as a function of its SNR for $\alpha=0.8$ and different values of $\text{SNR}_{PU}$. In this scenario, we observe a noticeable improvement by using improper signaling, with relative gains in the considered SNR regime of up to 9, 23 and 56 \% for $\text{SNR}_{PU}=10$, 20 and 30 dB, respectively. Fig. \ref{fig:ICalpha05} shows the achievable rate of the SU for $\alpha=0.5$. In this case, the improvement is significantly higher than in the previous scenario, achieving a relative rate improvement up to 292, 297 and 256 \% for $\text{SNR}_{PU}=10$, 20 and 30 dB, respectively. This is because the rate requirement for the PU is less stringent, hence providing the SU with more possibilities for improving its rate. Furthermore, a loading factor of $\alpha=0.5$ permits the PU to achieve its requirement only with the real or imaginary part of the transmitted symbol, which, according to Theorem \ref{le:ImpropCond}, makes improper signaling beneficial in all channel realizations and allows the SU to transmit with its maximum power.

\begin{figure}[t]
\centering
\includegraphics[width=0.85\columnwidth]{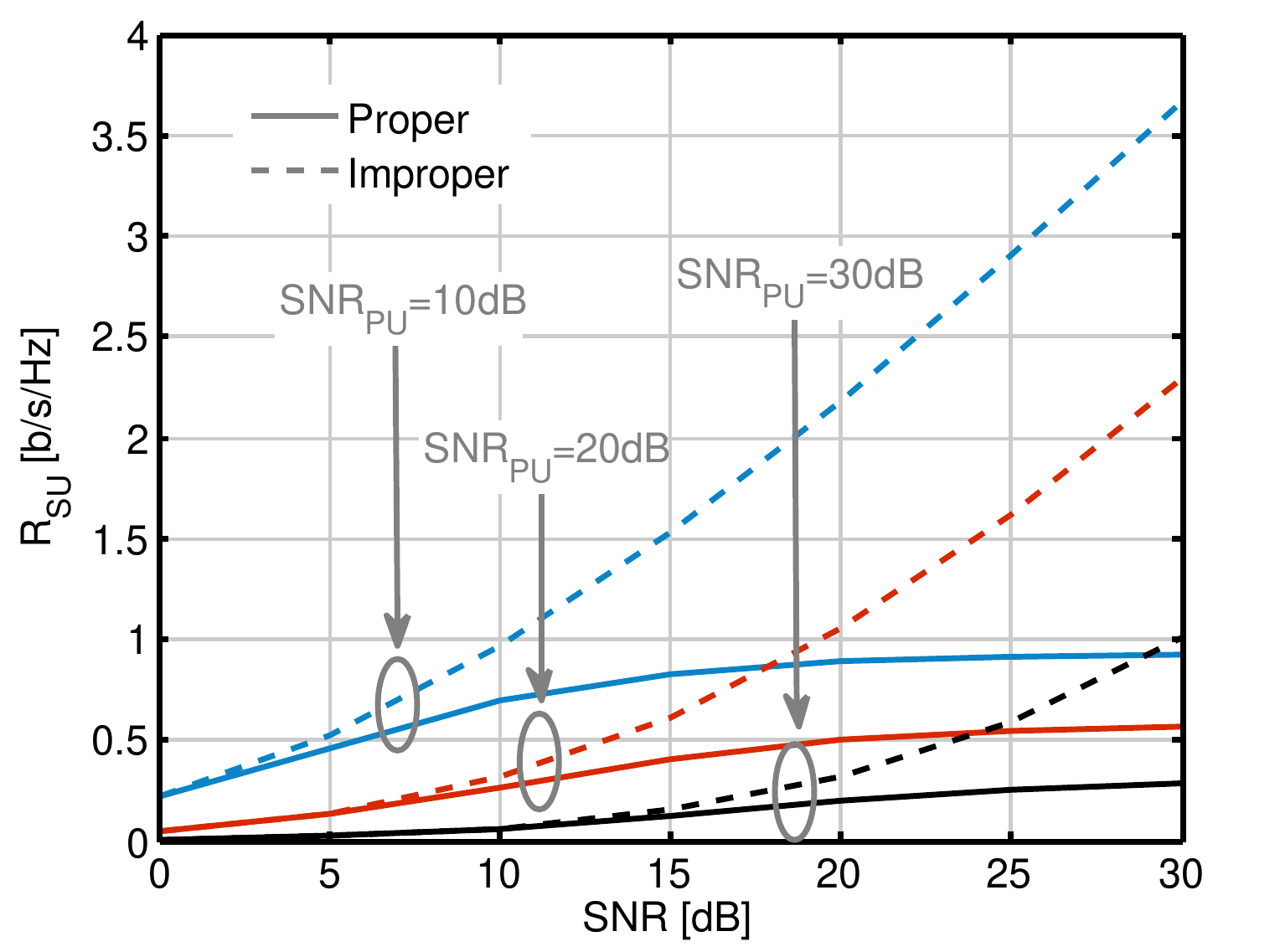}
\caption{Achievable rate of the SU as a function of its SNR, for different SNRs of the PU and $\alpha=0.5$.}
\label{fig:ICalpha05}
\end{figure}

\section{Conclusion}\label{sec:conclusions}
In this letter we have analyzed the benefits of improper Gaussian signaling for underlay CR. To this end, we have considered a basic yet illustrative scenario, the two-user SISO-IC, and derived the achievable rate of the SU as a function of its circularity coefficient when the PU has a rate constraint. This expression allows us to obtain a simple and insightful condition to determine when improper signaling is beneficial, as well as the optimal circularity coefficient in those cases. Our numerical results indicate that the SU may benefit from the transmission of improper signals especially when the PU is not highly loaded. Further lines of research include the design of practical constraints for the SU taking account of impropriety (e.g., interference temperature), as well as the extension of these results to other CR scenarios.

\appendices

\section{Proof of Lemma \ref{prop:Rsu}}\label{app:A}
Equating \eqref{eq:RpuImprop} to $\alpha R_{PU}(0,0)$ we have
\begin{equation}\label{eq:appA1}
    \frac{\left(\frac{q\left|g\right|^2+p\left|h\right|^2+\sigma^2}{q\left|g\right|^2}\right)^2-\kappa^2}{\left(\frac{q\left|g\right|^2+\sigma^2}{q\left|g\right|^2}\right)^2-\kappa^2}=2^{2\alpha R_{PU}(0,0)} \; .
\end{equation}
After some manipulations, the foregoing expression yields the second-order equation
\begin{equation}\label{eq:appA2}
    q^2=\frac{\sigma^2\left[\left(\frac{\gamma(2)}{\gamma(2\alpha)}-1\right)\sigma^2-2\gg\beta q\right]}{\left|g\right|^4\left(1-\kappa^2\right)} \; ,
\end{equation}
whose solution is given by \eqref{eq:qsuProp1}. On the other hand, using \eqref{eq:Rsu} we have
\begin{equation}\label{eq:appA3}
    2^{2R_{SU}}=\frac{\left|f\right|^4}{\tilde{\sigma}^4}\left(1-\kappa^2\right)q^2+\frac{2\fg}{\tilde{\sigma}^2}q+1 \; .
\end{equation}
Finally, \eqref{eq:RsuProp1} is obtained by substituting \eqref{eq:appA2} in \eqref{eq:appA3}, which concludes the proof.

\section{Proof of Theorem \ref{le:ImpropCond}}\label{app:B}
First, it is clear that $q(\kappa)$ is an increasing function. This can be readily observed by noticing that an improper interference increases the rate of the PU and, consequently, tolerates a higher amount of interference. This property can also be noticed by analyzing the derivative of $q(\kappa)$ with respect to $\kappa$, which yields the same conclusion and we omit due to lack of space. On the other hand, $R_{SU}(\kappa)$ depends on $\kappa$ only through its impact on $q(\kappa)$. Since the term within the logarithm in \eqref{eq:RsuProp1} is linear in terms of $q(\kappa)$, it will increase with $q(\kappa)$ (and, consequently, with $\kappa$) as long as it has a positive slope, i.e.,
\begin{equation}
    \frac{2\left|f\right|^2}{\tilde{\sigma}^2}\left(1-\beta\frac{\left|f\right|^2\sigma^2}{\left|g\right|^2\tilde{\sigma}^2}\right)>0 \; .
\end{equation}
After some manipulations of this expression, we obtain condition \eqref{eq:ImpropCond}. Now assume that condition \eqref{eq:ImpropCond} holds (i.e., improper signaling is beneficial) and let $Q$ be the power budget. Since $R_{SU}(\kappa)$ increases with $\kappa$, the optimal value of $\kappa$ can be set to 1 if the resulting transmit power is below the power budget, thus obtaining the first case in \eqref{eq:kappaOpt}. Otherwise, maximum power transmission is allowed by selecting $\kappa$ accordingly, i.e, $q(\kappa)=Q$, which yields the second case. This concludes the proof.

\bibliographystyle{IEEEtran}
\bibliography{Letter2014}

\begin{thebibliography}{10}
\providecommand{\url}[1]{#1}
\csname url@samestyle\endcsname
\providecommand{\newblock}{\relax}
\providecommand{\bibinfo}[2]{#2}
\providecommand{\BIBentrySTDinterwordspacing}{\spaceskip=0pt\relax}
\providecommand{\BIBentryALTinterwordstretchfactor}{4}
\providecommand{\BIBentryALTinterwordspacing}{\spaceskip=\fontdimen2\font plus
\BIBentryALTinterwordstretchfactor\fontdimen3\font minus
  \fontdimen4\font\relax}
\providecommand{\BIBforeignlanguage}[2]{{%
\expandafter\ifx\csname l@#1\endcsname\relax
\typeout{** WARNING: IEEEtran.bst: No hyphenation pattern has been}%
\typeout{** loaded for the language `#1'. Using the pattern for}%
\typeout{** the default language instead.}%
\else
\language=\csname l@#1\endcsname
\fi
#2}}
\providecommand{\BIBdecl}{\relax}
\BIBdecl

\bibitem{Cover}
T.~Cover and J.~Thomas, \emph{{Elements of Information Theory}}.\hskip 1em plus
  0.5em minus 0.4em\relax John Wiley, 2006.

\bibitem{Ho2012}
Z.~K. Ho and E.~Jorswieck, ``{Improper Gaussian Signaling on the Two-User SISO
  Interference Channel},'' \emph{IEEE Transactions on Wireless Communications},
  vol.~11, no.~9, pp. 3194--3203, Sep. 2012.

\bibitem{Zeng2013}
Y.~Zeng, C.~M. Yetis, E.~Gunawan, Y.~L. Guan, and R.~Zhang, ``{Transmit
  Optimization With Improper Gaussian Signaling for Interference Channels},''
  \emph{IEEE Transactions on Signal Processing}, vol.~61, no.~11, pp.
  2899--2913, Jun. 2013.

\bibitem{Cadambe2010}
V.~R. Cadambe, S.~A. Jafar, and C.~Wang, ``{Interference Alignment With
  Asymmetric Complex Signaling—Settling the H\o st-Madsen–Nosratinia
  Conjecture},'' \emph{IEEE Transactions on Information Theory}, vol.~56,
  no.~9, pp. 4552--4565, Sep. 2010.

\bibitem{ICC2013}
C.~Lameiro and I.~Santamar\'{\i}a, ``{Degrees-of-freedom for the 4-user SISO
  interference channel with improper signaling},'' in \emph{2013 IEEE
  International Conference on Communications (ICC)}.\hskip 1em plus 0.5em minus
  0.4em\relax Budapest, Hungary: IEEE, Jun. 2013, pp. 3053--3057.

\bibitem{Mitola1999}
J.~Mitola and G.~Maguire, ``{Cognitive radio: making software radios more
  personal},'' \emph{IEEE Personal Communications}, vol.~6, no.~4, pp. 13--18,
  1999.

\bibitem{Goldsmith2009}
A.~Goldsmith, S.~Jafar, I.~Maric, and S.~Srinivasa, ``{Breaking Spectrum
  Gridlock With Cognitive Radios: An Information Theoretic Perspective},''
  \emph{Proceedings of the IEEE}, vol.~97, no.~5, pp. 894--914, May 2009.

\bibitem{Yang2013}
Y.~Yang, G.~Scutari, P.~Song, and D.~P. Palomar, ``{Robust MIMO Cognitive Radio
  Systems Under Interference Temperature Constraints},'' \emph{IEEE Journal on
  Selected Areas in Communications}, vol.~31, no.~11, pp. 2465--2482, Nov.
  2013.

\bibitem{Cumanan2012}
K.~Cumanan, R.~Zhang, and S.~Lambotharan, ``{A New Design Paradigm for MIMO
  Cognitive Radio with Primary User Rate Constraint},'' \emph{IEEE
  Communications Letters}, vol.~16, no.~5, pp. 706--709, May 2012.

\bibitem{LameiroWSA2014}
C.~Lameiro, W.~Utschick, and I.~Santamar\'{\i}a, ``{Spatial Shaping and
  Precoding Design for Underlay MIMO Interference Channels},'' in \emph{2014
  20th International ITG Workshop on Smart Antennas (WSA)}, Erlangen, Germany,
  2014, pp. 1--8.

\bibitem{Jorswieck2011}
E.~A. Jorswieck and J.~Lv, ``{Spatial shaping in cognitive MIMO MAC with coded
  legacy transmission},'' in \emph{2011 IEEE 12th International Workshop on
  Signal Processing Advances in Wireless Communications}.\hskip 1em plus 0.5em
  minus 0.4em\relax IEEE, Jun. 2011, pp. 451--455.

\bibitem{Schreier2010}
P.~Schreier and L.~Scharf, \emph{{Statistical Signal Processing of
  Complex-Valued Data: The Theory of Improper and Noncircular Signals}}.\hskip
  1em plus 0.5em minus 0.4em\relax Cambridge, U.K.: Cambridge Univ. Press,
  2010.

\end{thebibliography}

\end{document}